\pdfoutput=1
\documentclass[sigconf]{acmart}
\usepackage{graphicx}
\usepackage{textcomp}
\usepackage{amsmath,amsthm,amsfonts}
\usepackage{thm-restate}
\usepackage{bbm}
\usepackage{xspace}
\usepackage{tikz}
\usepackage{color}
\usepackage{hyperref}
\usepackage{physics}
\usepackage[inline]{enumitem}

\usepackage{algpseudocode}
\usepackage{algorithm}
\usepackage{subcaption}
\usepackage{pdfpages}

\usepackage{hyperref}



\newtheorem{claim}{Claim}

\usepackage{tikz}
\usetikzlibrary{positioning, fit}
\usetikzlibrary{svg.path}
\usetikzlibrary{shapes.geometric}
\usetikzlibrary{patterns}
\usetikzlibrary{arrows.meta} %

\setenumerate[1]{label={(\arabic*)}}

\definecolor{blue}{RGB}{0,50,200}
\definecolor{magenta}{RGB}{255,0,255}
\hypersetup{colorlinks,linkcolor=blue,urlcolor=blue,citecolor=magenta}

\usetikzlibrary{positioning}

\newcommand{\prm}{\textsc{Circles}}

\newcommand{\N}{\mathbb{N}}

\usepackage[nolist, nohyperlinks]{acronym}
\acrodef{spp}[SPP]{standard population protocol}
\acrodef{ppus}[PPUS]{population protocol with unordered data and storage}
\acrodef{crn}[CRN]{chemical reaction network}

\title[Relative Majority with Cubic State Complexity]{
  Brief Announcement: Minimizing Energy Solves Relative Majority with a Cubic Number of States in Population Protocols
}

\author{Tom-Lukas Breitkopf}
\affiliation{%
    \institution{TU Berlin}
    \city{Berlin}
    \country{Germany}}
\email{t.breitkopf@tu-berlin.de}
\orcid{0009-0008-2875-1945}

\author{Julien Dallot}
\affiliation{%
  \institution{TU Berlin}
  \city{Berlin}
  \country{Germany}}
\email{judafa@protonmail.com}
\orcid{0009-0008-1286-1373}

\author{Antoine El-Hayek}
\affiliation{%
  \institution{Institute of Science and Technology Austria}
  \city{Klosterneuburg}
  \country{Austria}}
\email{antoine.el-hayek@ist.ac.at}
\orcid{0000-0003-4268-7368}

\author{Stefan Schmid}
\affiliation{%
 \institution{TU Berlin}
 \city{Berlin}
 \country{Germany}}\email{stefan.schmid@tu-berlin.de}
 \orcid{0000-0002-7798-1711}

\acmYear{2025}\copyrightyear{2025}
\acmConference[PODC '25]{ACM Symposium on Principles of Distributed Computing}{June 16--20, 2025}{Huatulco, Mexico}
\acmBooktitle{ACM Symposium on Principles of Distributed Computing (PODC '25), June 16--20, 2025, Huatulco, Mexico}
\acmDOI{10.1145/3732772.3733512}
\acmISBN{979-8-4007-1885-4/25/06}

\begin{document}

\begin{abstract}
  This paper revisits a fundamental distributed computing problem in the population protocol model. 
  Provided $n$ agents each starting with an input color in $[k]$, the relative majority problem asks to find the predominant color.
  In the population protocol model, at each time step, a scheduler selects two agents that first learn each other's states and then update their states based on what they learned.

  We present the \prm{} protocol that solves the relative majority problem with $k^3$ states. It is always-correct under weakly fair scheduling.
  Not only does it improve upon the best known upper bound of $O(k^7)$, but it also shows a strikingly simpler design inspired by energy minimization in chemical settings.
\end{abstract}

\maketitle

\section{Introduction}
\label{sec:intro}

Population protocols model a computation distributed across a population of $n$ all-identical agents that interact in a chaotic, unpredictable manner.
The model was introduced by Angluin et al.~in 2006~\cite{angluin2006networksoffinitestatesensors} to model a network of small sensors; since then, it attracted a growing attention \cite{elsasser2018recentresultsinpp, alistarh2018recentalgorithmicadvancesinpp, bankhamer2022ppforexactpluralityconsensus, czerner2023lowerboundsonthestatecomplexityofpp} for its broad applications ranging from dynamics in social groups \cite{becchetti2015pluralityconsensus} to chemical reactions \cite{doty2014timingincrns, natale2019necessarymemorycomputeplurality} as well as its theoretical interest~\cite{s00446-007-0040-2, lipton1976}.

\textbf{Model.}
In this paper, we focus on the relative majority problem.
In this problem, each agent is initially assigned one \emph{input color} in $0, 1 \dots k-1$ and the goal is to find the color with the greatest support (we assume no ties).
As part of the problem, a \emph{scheduler} also specifies an infinite sequence of \emph{pairwise interactions} between the agents.
Then the protocol runs: pairs of agents interact one after the other according to the planned schedule; when two agents interact, they both 1) learn the \emph{state} of the other agent and 2) update their current state as the protocol specifies.
Two agents with the same state are perfectly identical: after it interacted, an agent's new state only depends on its previous own state and on the state of the other agent it just interacted with.

\begin{definition}[Configuration]
  A configuration is a complete description of the population at a given time.
  As agents with the same state are identical, we define a configuration as the multiset that contains all the states of the population.
\end{definition}

In the context of population protocols, we say that a protocol solves the problem if, for all possible input color assignments and all possible sequences of interactions, every agent eventually outputs the correct majority color, forever.
However, an unconstrained scheduler makes the problem trivially impossible (by isolating some agents for instance), we therefore assume that the scheduler is \emph{weakly fair}:

\begin{definition}[Weakly Fair Scheduler]
  A weakly fair scheduler produces interaction schedules where each possible interaction pair happens infinitely often.
\end{definition}

\textbf{Contribution.}
We present \prm{}, an always-correct protocol that solves the relative majority problem under a weakly fair scheduler.
We designed \prm{} with an emphasis on \emph{state complexity}, which is the number of different states an agent can have.
\prm{} has a state complexity of $k^3$, which improves upon the best known upper bound of $O(k^{7})$ \cite{gasieniec2017deterministicppforexactmajorityandpluarlity} and narrows the gap with the best known lower bound of $\Omega(k^{2})$ \cite{natale2019necessarymemorycomputeplurality}.
The emphasis on state complexity is motivated by applications where the memory space per agent is severely limited: tiny sensors in a network~\cite{angluin2006networksoffinitestatesensors} or molecules in chemical applications~\cite{czerner2023lowerboundsonthestatecomplexityofpp, alistarh2018recentalgorithmicadvancesinpp}.
\prm{} is always correct under a weakly fair scheduler and shows an elegant design inspired by energy minimization in chemical settings.

\textbf{Notations}
We define a few notations used in the protocol's definition and proofs.

\textit{multisets. }
For sets $S$ and $T$ we write $S^{T}$ to denote the set of functions $f : T \to S$.
If $T$ is finite we call the elements of $\mathbb{N}^T$ \emph{multisets} over $T$.
In this paper, the subset $\subseteq$, the union $\cup$ and the set substraction $\setminus$ operations are systematically generalized to multisets.

\textit{remainder. }
For $x \in \mathbb{Z}$ and $p \in \mathbb{N}^*$, we define $x \bmod p$ as the remainder of the euclidean division of $x$ by $p$.
Note that this is a number in $\mathbb N$, not $\mathbb Z/n\mathbb Z$.

\textit{ranges. }
Let $x, y \in \mathbb{N}$ such that $x \le y$.
$[x, y]$ denotes the set $\{x, x+1, \dots, y-1, y\}$, and $(x, y)$ denotes the set $\{x+1, x+2, \dots, y-2, y-1\}$.

\textit{modulo ranges.}
Let $x, y \in \mathbb{N}$.
We define modulo ranges:  $[x, y]_p$ denotes the set $\{x \bmod p, (x+1) \bmod p, \dots, (x+(y-x) \bmod p - 1) \bmod p, (x + (y-x) \bmod p) \bmod p\}$, and $(x, y)_p$ denotes the set $\{(x+1) \bmod p , (x+2) \bmod p, \dots, (x + (y-x) \bmod p - 2) \bmod p, (x + (y - x) \bmod p - 1) \bmod p\}$.
For instance, $[2,7]_{10} = \{2, 3, 4, 5, 6, 7\}$ and $(8, 3)_{10} = \{9, 0, 1, 2\}$.

\textit{bra-ket. }
We abuse the bra-ket notation, frequently used in quantum mechanics, to note ordered pairs.
For $i, j \in \N$, we write $\braket{i}{j}$ simply to distinguish the different roles of $i$ and $j$.
For an agent storing a bra-ket $\braket{i}{j}$, we refer to $i$ as its bra and $j$ as its ket.

\section{The \prm{} protocol}
\label{sec:protocol_overview}

We present thereafter the \prm{} protocol that runs in every agents: its set of states, input function (to convert the input color into one of the protocol's states), output function (to ask an agent what the majority color is) and the transition function (to specify how two agents update their states when they interact).
\begin{itemize}
  \label{prot:circle}
\item
  \textbf{States:} The set of states $Q$ contains every triples $(i, j, o) \in [0, k-1]^{3}$.
  In the remainder of this paper we will use the bra-ket notation $\braket{i}{j}$ to refer to the first two numbers of the triple, \emph{bra} refers to $i$ and \emph{ket} refers to $j$, while \emph{out} refers to $o$. 
  \item \textbf{Input:} each agent is initialized with $\braket{i}{i}$ and $\emph{out} = i$, where $i$ is the input color of the agent.
  \item \textbf{Output:} return $\emph{out}$
  \item \textbf{Transition function:} We define \emph{weights} for each bra-ket $\braket{i}{j}$ as follows:
  \begin{equation*}
  	w(\braket{i}{j}) =
  	\begin{cases}
  		k & \textit{if } i = j\\
  		(j - i) \bmod k& \textit{otherwise}
  	\end{cases}
  \end{equation*}

  Two agents $a$ and $b$ that interact perform two successive operations:
\begin{enumerate}
	\item $a$ and $b$ exchange their kets in case this strictly decreases the minimum weight of their two bra-kets.
    \item If either $a$ or $b$ is of the form $\braket{i}$ for some $i \in [k]$, set $\emph{out}_a=\emph{out}_b=i$.
\end{enumerate}
\end{itemize}

\section{Proof of correctness}
\label{sec:protocol_proof}

We prove the protocol's correctness in Theorems \ref{theorem:stabilization} and \ref{theorem:correctness}, Theorem \ref{theorem:correctness} directly derives from Lemma \ref{lem:majowins}.
We beforehand introduce Greedy Independent sets, a construction on the input colors used to prove the protocol's correctness, as well as two preliminary lemmas \ref{lemma:majority_color} and \ref{lemma:global_braket_invariant}.

\begin{definition}[Greedy Independent Sets]
  \label{def:greedy_independent_sets}
  Consider the multiset of input colors to our protocol.
  We partition this multiset into sets $G_1, G_2, \dots, G_q$ as follows: store in $G_1$ as many inputs as possible as long as $G_1$ does not contain two equal colors; then apply the same on the remaining inputs to obtain $G_2$, $G_3$, and so on.
\end{definition}

\begin{restatable}{lemma}{majoritycolor}
  \label{lemma:majority_color}
  (Majority Color).
  Assume that there exists a unique color $\mu$ in relative majority, then it holds that $G_q = \{\mu\}$ and there is no $j \neq \mu$ and $p \in [1, q]$ such that $G_p = \{j\}$.
\end{restatable}

\begin{proof}
  Each time a set $G_p$ is constructed according to Definition \ref{def:greedy_independent_sets}, any color whose count is not zero yet is added into $G_p$ and its count is decremented by one.
  The color $\mu$ appears in the population strictly more than any other, it therefore holds that $\forall p \in [1, q]$, $\mu \in G_p$.
  Let now $i \in [0, k-1]$ be a color contained in a set $G_p$ for some $p \in [1, q]$.
  It holds that $\forall l \le p$, $i \in G_l$ because color $i$ is available to populated $G_p$ so $i$ was also available when $G_l$ was filled earlier.
  It therefore cannot be that a color $j \neq \mu$ is contained in $G_q$ as $j$ would be contained in all sets $G_1 \dots G_q$ and thus $j$ would also be in relative majority, a contradiction.
\end{proof}

\begin{restatable}{lemma}{invariant}
  \label{lemma:global_braket_invariant}
  (Global Braket Invariant).
  In every configuration and for all $i \in [0, k-1]$, the number of bras $\bra{i}$ and the number of kets $\ket{i}$ are equal.
\end{restatable}

\begin{proof}
  Every agent is initialized with bra-ket $\braket{i}{i}$ for some $i \in [0, k-1]$ and the claim initially holds.
  Agents subsequently only ever update their bra-ket, by exchanging kets among each other.
  The overall number of bras and kets in the population therefore does not change during a computation.
\end{proof}

\begin{restatable}{theorem}{stabilization}
  \label{theorem:stabilization}
  (Stabilization).
  The agents exchange their kets a finite number of times.
\end{restatable}
\begin{proof}
  We call $\omega$ the smallest ordinal number greater than all the integers.
  We prove the claim by exhibiting a non-negative quantity that strictly decreases at each ket exchange.
  Given a configuration $C$, let $w_1(C), w_2(C) \dots w_n(C)$ be the bra-ket's weights of each agent sorted in increasing order.
  Define
  \begin{align*}
    g(C) =  \omega^{n-1} \cdot w_1(C) + \omega^{n-2} \cdot w_2(C)  + \dots + \omega \cdot w_{n-1} + 1 \cdot w_{n}
  \end{align*}
  Assume that two agents exchange their kets.
  Let $p$ be the lowest index such that $w_p(C)$ changes before and after the ket exchange.
  By design of the protocol, $w_p(C)$ strictly decreases.
  This implies that $g$ strictly decreases when two agents exchange their kets.
  As an ordinal number cannot decrease infinitely many times, the number of ket exchanges is therefore finite.
\end{proof}

\noindent
We define the following special sets of bra-kets only to formulate Lemma~\ref{lem:majowins}.
\begin{definition}[Circle Bra-ket Sets]\label{def:braketsets}
    For a given greedy independent set $G_p$ with $p \in [1,q]$ (Definition~\ref{def:greedy_independent_sets}), let $g_0, g_1, \dots ,g_m$ be the elements of $G_p$ sorted in increasing order and define $$f(G_p) = \{\braket{g_0}{g_1}, \braket{g_1}{g_2}, \dots , \braket{g_m}{g_0}\}$$
\end{definition}

\begin{lemma}\label{lem:majowins}
    After Stabilization (Theorem~\ref{theorem:stabilization}), let $\mathcal C$ be the multiset of bra-kets of the agents. We have that:
    $$\mathcal{C} = \bigcup_{p=1 \dots q} f(G_p)$$
\end{lemma}

\begin{proof}

 We prove the predicate $H(r)$ by induction on $r \in [0, q]$:
  \begin{align}
    \bigcup_{p = 1 \dots r} f(G_p) \subseteq \mathcal{C} \tag{$H(r)$}
  \end{align}

  \noindent
  The base case for $r=0$ is trivial. 
  Let $r \in [0, q-1]$, we assume that $H(r)$ holds and show that $H(r+~1)$ also holds.
  We define the subconfiguration $\mathcal{C}[r+1] = \mathcal{C} \setminus \cup_{p = 1 \dots r} f(G_p)$.

  Case 1: $\cup_{p=r+1}^q G_p$ contains only elements from one color.
  Let $i$ be that color.
    
    Then for any other color $j \neq i$, there are at most $r$ many bras $\bra{j}$ and as many kets $\ket{j}$, which are all included in $\cup_{p = 1 \dots r} f(G_p)$. Thus all agents in $\mathcal C[r+1]$ are of the form $\braket{i}$. Since
    $\{\braket{i}\} = f(G_{r+1})$, we have $ \cup_{p = 1 \dots r+1} f(G_p) \subseteq \mathcal{C}$.

  Case 2: There are at least two different colors in $G_{r+1}$.
  
  We note $g_0, g_1, \dots, g_m$ the elements of $G_{r+1}$ sorted in increasing order.
  Let $l \in [0, m]$. To lighten notations, we will mean $l + s \bmod m+1$ each time we write $l+s$ in the remainder of this proof.
  We prove that, if there is no agent with bra-ket $\braket{g_l}{g_{l+1}}$ in $C[r+1]$, then there exist two agents whose interaction creates that bra-ket, a contradiction with the stability hypothesis.
  First notice that $\bra{g_l}$ and $\ket{g_{l+1}}$ are in $\mathcal{C}[r+1]$.
  Indeed, note that there are at least $r+1$ many $\bra{g_l}$ and $r+1$ many $\bra{g_{l+1}}$ in $\mathcal C$, as there are at least $r+1$ many agents with color $g_l$ and $g_{l+1}$ initially.
  By Theorem~\ref{lemma:global_braket_invariant}, this means we have at least $r+1$ many $\ket{g_{l+1}}$ in $\mathcal C$.
  Because each $f_p$ for $p \le r$ contains exactly one $\bra{g_l}$ and one $\ket{g_{l+1}}$, we have that both $\bra{g_l}$ and $\ket{g_{l+1}}$ are in $\mathcal{C}[r+1]$.
  
  Assuming by contradiction that there is no agent with bra-ket $\braket{g_l}{g_{l+1}}$ in $\mathcal{C}[r+1]$, then there exists an agent with bra-ket $\braket{g_l}{j}$ and an agent with bra-ket $\braket{i}{g_{l+1}}$ in $\mathcal{C}[r+1]$ for some $i$ and $j$.
  We show that those two agents exchange their kets if they interact.

  \begin{claim}\label{claim:notininterval}
      $i, j \notin (g_l, g_{l+1})_m$
  \end{claim}
  \begin{proof}
      By contradiction, assume that $i$ is in $(g_l, g_{l+1})_m$. 
      A $\bra{i}$ in $\mathcal C[r+1]$ indicates that $i$ had initially at least $r+1$ agents supporting it, as by contradiction if it wasn't the case, all the $\bra{i}$ would have been in $\cup_{p=1}^r f(G_p)$.
      By construction of $G_{r+1}$, we must have that $i$ is in $G_{r+1}$, and thus, $g_l$ and $g_{l+1}$ are not consecutive in the ordered list of $G_{r+1}$, a contradiction.
      The case $j \in (g_l, g_{l+1})_m$ is symmetric.
  \end{proof}
  If $i \neq g_{l+1}$, it holds by Claim~\ref
{claim:notininterval} that
  \begin{align*}
    w(\braket{g_l}{g_{l+1}})= (g_{l+1} - g_l) \bmod k < (g_{l+1} - i) \bmod k = w(\braket{i}{g_{l+1}})
  \end{align*}
  Otherwise, if $i=g_{l+1}$:
  \begin{align*}
    w(\braket{g_l}{g_{l+1}})\quad=\quad (g_{l+1} - g_l) \bmod k \quad < \quad k \quad =\quad w(\braket{i}{g_{l+1}})
  \end{align*}

  Similarly, if $j \neq g_l$, it holds by Claim~\ref
{claim:notininterval} that
  \begin{align*}
    w(\braket{g_l}{g_{l+1}}) \ = \ (g_{l+1} - g_l) \bmod k  \ <  \ (j - g_{l}) \bmod k  \ = \ w(\braket{g_l}{j})
  \end{align*}
  Otherwise, if $j=g_{l}$:
  \begin{align*}
    w(\braket{g_l}{g_{l+1}})\quad=\quad (g_{l+1} - g_l) \bmod k \quad < \quad k \quad =\quad w(\braket{g_l}{j})
  \end{align*}
  
  This proves that exchanging kets between $\braket{g_l}{j}$ and $\braket{i}{g_{l+1}}$ reduces the minimum weight.  
  An interaction between the two agents eventually happens as the scheduler is weakly fair, this interaction would therefore trigger a ket exchange, which is in contradiction with the stability hypothesis: we deduce that $\braket{g_l}{g_{l+1}} \in \mathcal{C}[r+1]$, which implies $f(G_{r+1}) \subseteq \mathcal{C}[r+1]$ and therefore $H(r+1)$ holds.
  
  We proved by induction that $H(q)$ holds:
  \begin{align*}
    \bigcup_{p = 1 \dots q} f(G_p) \subseteq \mathcal{C}
  \end{align*}
  As $|\cup_{p = 1 \dots q} f(G_p)| = |\mathcal{C}|$ we can rewrite $H(q)$ as an equality and the claim holds. 
\end{proof}

\begin{restatable}{theorem}{correctness}
  \label{theorem:correctness}
  (Correctness).
  Assume that there exists a unique color $\mu$ in relative majority.
  In the \prm{} protocol, all agents eventually output $\mu$ under a weakly fair scheduler.
\end{restatable}

\begin{proof}
    By Lemma~\ref{lem:majowins} and Lemma~\ref{lemma:majority_color}, after Stabilization (Theorem~\ref{theorem:stabilization}), since we assumed that there is only one majority color, there exists at least one agent in bra-ket $\braket{\mu}$ and none in bra-ket $\braket{j}$ for $j \neq \mu$.
    The agent(s) with bra-ket $\braket{\mu}{\mu}$ will transmit their output color to the rest of the population and the claim follows.
\end{proof}

\section{Extensions}
We plan to expand the functionalities of \prm{} to handle ties and/or to operate in an unordered setting.
Those extensions will be published as a more complete version of the present work.

\textbf{Handling ties.}
We can extend \prm{} to handle ties in multiple ways.
For instance, all agents can indicate a tie with a special state (tie report), agree on one unique winning color (tie break), or output their own color if their input color wins while the losers output any winning color (tie share).
It is possible to implement all those ways to handle ties by adding simple extra-layer protocols on top of \prm{} while keeping the state complexity at $O(k^{3})$.

\textbf{Unordered setting.}
The \prm{} protocol, which we introduce in this work, relies on numerical representations of the colors in order to compute some kind of distance between them.
It can be adapted to the unordered setting (in which agents are only able to compare colors for equality and memorize them) using~$O(k^4)$ states.
For that we propose a new protocol to generate an ordering between colors using~$O(k^2)$ states.
Adapting a protocol proposed in~\cite{shukai2012selfstabilizingleaderelection} we perform leader-election between all agents of the same color (using the asymmetry of interactions) and have the leaders increment a numeric label every time they meet another leader with the same label. The non-leaders simply copy the label of their leader.
Similar to~\cite{natale2019necessarymemorycomputeplurality} we then combine the ordering protocol with \prm{} by re-initializing agents of some color whenever their numeric label (representing that color) changes. For that we need to put agents into special states in which they wait to undo changes they previously made to the population until they are ``consistent'' again and ready to be re-initialized.
In order to use as few states as possible we do not explicitly store the output of the ordering protocol, but write it directly to the bra of an agent.

\begin{acks}
    This project has received funding from the European Research Council (ERC) under the European Union's Horizon 2020 research and innovation programme (MoDynStruct, No. 101019564)  \includegraphics[width=0.9cm]{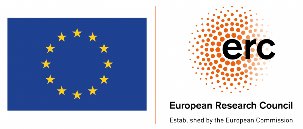} and the Austrian Science Fund (FWF) grant \href{https://www.doi.org/10.55776/I5982}{DOI 10.55776/I5982}, and grant \href{https://www.doi.org/10.55776/P33775}{DOI 10.55776/P33775} with additional funding from the netidee SCIENCE Stiftung, 2020–2024 and the German Research Foundation (DFG), grant 470029389 (FlexNets).
\end{acks}

\bibliographystyle{ACM-Reference-Format}
\bibliography{references}

\end{document}